\documentclass[11pt,a4paper]{article}
\usepackage[english]{babel}
\usepackage{amsmath,amssymb,amsthm,epsfig,color,graphicx}
\usepackage[latin1]{inputenc}

\usepackage{tikz, tkz-euclide}
\usetikzlibrary{calc}
\usetikzlibrary{decorations.markings}

\usepackage{bbm}
  \newtheorem*{theorem*}        {Theorem}
	\newtheorem*{conjecture*}   {Conjecture}
  \newtheorem{theorem}           {Theorem}
  
  \newtheorem*{lemma*}          {Lemma}

  \newtheorem{proposition}      {Proposition}

\definecolor{Red}{cmyk}{0,1,1,0}

\definecolor{Blue}{cmyk}{1,1,0,0}

\newcommand{\ba}{\begin{array}}
\newcommand{\ea}{\end{array}}
\newcommand{\be}{\begin{equation}}
\newcommand{\ee}{\end{equation}}
\newcommand{\ben}{\begin{enumerate}}
\newcommand{\een}{\end{enumerate}}

  \def\d{\mathop{\textrm{\rm d}}\nolimits}                  
  \def\conv{\mathop{\textrm{\rm conv}}\nolimits}            
  \def\exp{\mathop{\textrm{\rm exp}}\nolimits}              
	\def\arctanh{\mathop{\textrm{\rm arctanh}}\nolimits}            

\let\x=\xi

\begin{document}

\title{Stability of the Phase Transition of Critical-Field Ising Model on Cayley trees under Inhomogeneous External Fields}
\author{
Rodrigo Bissacot\thanks{Partially supported by the Dutch stochastics cluster STAR (Stochastics - Theoretical and Applied
Research), also supported by FAPESP Grants 11/16265-8, 2016/08518-7 and CNPq Grants 486819/2013-2, 312112/2015-7.}\\
\footnotesize{\texttt{rodrigo.bissacot@gmail.com}}\\
\footnotesize{Institute of Mathematics and Statistics - IME USP - University of S\~ao Paulo}
\\[0.3cm]
Eric Ossami Endo\thanks{Supported by FAPESP Grants 14/10637-9 and 15/14434-8.} \\
\footnotesize{\texttt{eric@ime.usp.br}}\\
\footnotesize{Institute of Mathematics and Statistics - IME USP - University of S\~ao Paulo}\\
\footnotesize{Johann Bernoulli Institute for Mathematics and Computer Science - University of Groningen}\\
\\
Aernout C. D. van Enter\\
\footnotesize{\texttt{a.c.d.van.enter@rug.nl}}\\
\footnotesize{Department of Mathematics}\\
\footnotesize{Johann Bernoulli Institute for Mathematics and Computer Science - University of Groningen}
}
\maketitle

\begin{abstract}
We consider the ferromagnetic Ising model with spatially dependent external fields on a Cayley tree, and we investigate the conditions for 
the existence of the phase transition for a class of external fields, asymptotically approaching a homogeneous critical external field. Our results extend earlier results by Rozikov and Ganikhodjaev.
\end{abstract}



\section{Introduction}

The ferromagnetic Ising model on a Cayley tree 
has been extensively studied. Some early treatments which were mathematically rigorous can be found in papers by Katsura and Takizawa \cite{KT} and by Preston \cite{Preston}. As opposed to the situation on $\mathbb{Z}^d$, the phase transition on Cayley trees can appear even when the model has a non-zero homogeneous external field, because these trees are non-amenable, see \cite{JS}. 

Preston analysed the phase transition using Markov chains. More precisely, he proved that a large class of  Gibbs measures can be written as 
Markov chains, and the set of translational invariant Gibbs measures contains one, two or three completely homogeneous Markov chains.  For more details, see \cite{Ge} and \cite{Rozi1}.

The main approach to show the phase transition of the Ising model on a Cayley tree uses the fact that we can restrict ourselves to this well-behaved class of Gibbs measures, to which (among others) each extremal Gibbs measure belongs; this is the class of probability measures named ``splitting Gibbs measures" or ``Markov chains" (see \cite{Ge, Rozi1} and \cite{Ku}). The advantage to work with this class of probability measures is the fact that, on Cayley trees, we have a notion of compatibility (see the definition in the next section),  which implies that such a probability measure is a Markov chain. By this approach Preston showed that there exists a positive critical value for homogeneous external fields depending on the temperature and the order of the Cayley tree $\Gamma^d$, indicated by $h_c(\beta, d)$, and a critical temperature, indicated by $\beta_c(d)$, such that, if $\beta\le \beta_c(d)$ or $|h|>h_c(\beta,d)$, there is no phase transition, and for $\beta>\beta_c(d)$ and $|h|\le h_c(\beta,d)$, the model undergoes a phase transition. Moreover, when $|h|=h_c(\beta,d)$, there exist exactly two homogeneous splitting Gibbs measures, while for $|h|<h_c(\beta,d)$ there exist exactly three homogeneous splitting Gibbs measures.

Technically, the advantage is that the study of Gibbs measures reduces to the study of a set of recurrence equations on ``boundary laws'' (which we sometimes also call ``boundary fields''). This reduction actually does not need any translation invariance; for some earlier work where non-translation-invariant versions play a role, see e.g. \cite{Rozi2}, and \cite{EEIK} and in particular the discussion in its Appendix.

When we assume inhomogeneous external fields, we want to consider perturbations of this critical field, with the perturbations decaying to zero when the distance to the root approaches infinity.

We consider this regime, as it is the most sensitive, and presents the closest equivalent of the question considered in \cite{BC,BCCP}. We want to look how fast the external fields can decay to the critical value in order to still be able to see a phase transition. 

Moreover, we remark that spin systems in decaying fields have been studied to model systems in traps, see e.g. \cite{CV1}.

On the lattice $\mathbb{Z}^d$, we know by the Peierls argument that in zero field there is a transition for 
$d\ge 2$. On the other hand, Lee and Yang \cite{LY, Rue} showed that the Ising model with non-zero homogeneous external fields has uniqueness for any temperature. Thus, the critical value in this case is $h_c=0$. 

When the external fields are depending on the site $h_i$ with $i\in \mathbb{Z}^d$ and they are decaying to zero, we have the following results. Bissacot and Cioletti \cite{BC} showed that the ferromagnetic Ising model on the lattice $\mathbb{Z}^d$ with external fields $(h_i)_{i\in \mathbb{Z}^d}$ undergoes a phase transition if the external fields $h_i$ are summable, i.e., $\sum_{i\in \mathbb{Z}^d}h_i<+\infty$. Moreover, Bissacot, Cassandro, Cioletti and Presutti \cite{BCCP} showed that, if the external fields are of the form $h_i=\lVert i\rVert^{-\alpha}$, then the ferromagnetic Ising model on the lattice $\mathbb{Z}^d$ with external fields $(h_i)_{i\in \mathbb{Z}^d}$ undergoes a phase transition when $\alpha>1$, and there is uniqueness at $0<\alpha<1$ for a small temperature. Afterwards, Cioletti and Vila \cite{CV} extended the uniqueness for all temperatures when $0<\alpha<1$. For $\alpha=1$ there are partial results, see \cite{BCCP}.

The purpose of this paper is to look at the ferromagnetic Ising model on the Cayley tree $\Gamma^d$ with spatially dependent external fields. As in \cite{BCCP}, the external fields are decaying to the critical value, which here is non-zero, $h_c=h_c(\beta,d)$. We will show that, if the external fields are of the form $h_n=-h_c-\epsilon_n$, with $\epsilon_n$ positive, decreasing, decaying to zero and satisfying the following  condition,
\begin{equation}\label{sum1}
\lim_{n\to \infty}\sum_{j=1}^n\left( \sum_{i=j}^n \epsilon_i \right)^2<\infty,
\end{equation}
then the model undergoes a phase transition at low temperature. 
On the other hand, when the condition is violated, and the sum diverges, we will obtain uniqueness of the Gibbs measure for the perturbed model.
 
Note that the above condition is substantially weaker than the one from \cite{BC}. In fact, the condition that we would get from the arguments of  \cite{BC} is
\begin{equation}\label{sum2}
\sum_{n\ge 1}d^n\epsilon_n<+\infty.
\end{equation}
It is easy to see that every sequence $(\epsilon_n)_{n\ge 1}$ satisfying (\ref{sum2}) also satisfies (\ref{sum1}).
We will say that the phase transition persists if the maximal measure $\mu^{+}$ is different from the minimal measure $\mu^{-}$.
We will consider in particular the persistence of the plus state (that is the positively magnetised state) in a negative critical field. We notice that this state is unstable (disappears) if we add a homogeneous negative external field; we ask the question what happens once we add a negative field decaying to zero. If the decay is fast the transition persists, if the decay is slow enough it may disappear. The condition mentioned above indicates the threshold between those two behaviours. 

Some other aspects of the Ising model in a critical field have been studied by Bleher et al \cite{BRSSZ}. The stability for decaying interactions of the Ising model in non-critical fields, and also for summable inhomogeneous fields as in \cite{BC}, has been considered by Ganikhodjaev \cite{Gan}.

\section{Definitions and notation}

Let $\Gamma^d=(V,L)$ be the Cayley tree of order $d$, i.e., a $d+1$-regular infinite tree.
For a fixed $x_0 \in V$, called the root, define generation $n$ by $W_n=\{x\in V:\ \d(x,x_0)=n\}$ and $V_n=\cup_{k=1}^n W_k$. For each $x\in W_n$, denote $S(x)=\{y\in W_{n+1}:\ \d(x,y)=1\}$
for the set of  children of $x$. Define also $L_n$ to be the edges of the subtree of $\Gamma^d$ restricted to the vertices $V_n$. For each $x\in V$, we denote the distance of $x$ from the root $x_0$ by $\lVert x\rVert=\d(x,x_0)$.

Consider the set $\Omega=\{-1,1\}^V$ of configurations on $V$. For each $\Lambda\subset V$, we define $\Omega_{\Lambda}=\{-1,1\}^{\Lambda}$ to be the set of configurations on $\Lambda$. We denote by $\sigma\in \Omega$ the configurations in $\Omega$, and by $\sigma_{\Lambda}\in \Omega_{\Lambda}$ the configurations in $\Omega_{\Lambda}$ (if the set is well-known and there is no chance of confusion, we will sometimes write $\sigma$ instead of $\sigma_{\Lambda}$.)

Let $\sigma \in \Omega$ and $J>0$. The \emph{Hamiltonian of the ferromagnetic Ising model on the volume $V_n$} is the following,
\begin{equation}
H^{0}_{n}(\sigma)=
-J\displaystyle\sum_{\langle x, y\rangle \in L_n}\sigma_x\sigma_y.
\end{equation}
where $\langle x,y\rangle$ is summed over  a set of  pairs of nearest-neighbor vertices.  
The \emph{Hamiltonian of the ferromagnetic Ising model with boundary condition $\eta$ on the volume $V_n$} is the following,
\begin{equation}
H^{\eta}_{n}(\sigma)=
-J\displaystyle\sum_{\langle x, y\rangle \in L_n}\sigma_x\sigma_y-J\displaystyle\sum_{\substack{\langle x, y\rangle \in L_{n+1}\\ x\in W_n\\ y\in W_{n+1}}}\sigma_x\eta_y.
\end{equation}
When $\eta_x=1$ for all $x\in V$, we say that $\eta$ is the \emph{plus boundary condition}, while $\eta_x=-1$ for all $x\in V$ is the \emph{minus boundary condition}, and $\eta_x=0$ for all $x\in V$ is the \emph{free boundary condition}.

Let $\bar{h}=(h_n)_{n\ge 1}$ be a real-valued sequence. We define the \emph{Hamiltonian of the ferromagnetic Ising model on the volume $V_n$ with boundary condition $\eta$ and with spatially dependent external fields $\bar{h}$} by
\begin{equation}\label{hamiltonian}
H^{\eta}_{n,\bar{h}}(\sigma)=
H^{\eta}_{n}(\sigma)  - \sum_{k=1}^n\sum_{x\in W_k}h_k \sigma_x.
\end{equation}

Given the \emph{inverse temperature} $\beta>0$, we define the \emph{Gibbs measure  on the volume $V_n$ with boundary condition $\eta$ and with spatially dependent external fields $\bar{h}$} by
\begin{equation}
\mu_{n,\beta,\bar{h}}^{\eta}(\sigma)=
\frac{\exp\{-\beta H^{\eta}_{n,\bar{h}}(\sigma)\}}{Z^{\eta}_{n,\beta,\bar{h}}},
\end{equation}
where $Z^{\eta}_{n,\beta,\bar{h}}$ is the partition function given by
\begin{equation}
Z^{\eta}_{n,\beta,\bar{h}}=\sum_{\sigma\in \Omega_{V_n}}\exp\{-\beta H^{\eta}_{n,\bar{h}}(\sigma)\}.
\end{equation}
We will write $\mu_{n,\beta,\bar{h}}^{+}$ when the boundary condition is plus, and $\mu_{n,\beta,\bar{h}}^{-}$ when it is minus. Define the set of Gibbs measures as the convex hull of the set of all weak limits of the Gibbs measures on the volumes $V_n$ with boundary conditions $\eta_n$ and with spatially dependent fields $\bar{h}$,
\begin{equation}
\mathcal{G}_{\beta}=\conv\left\{ \mu: \lim_{n\to +\infty}\mu_{m_n,\beta,\bar{h}}^{\eta_{m_n}}=\mu; \text{ for all increasing } (m_n)_{n\ge 1}, \text{ and } (\eta_n)_{n\ge 1} \right\}
\end{equation}
We say that the model undergoes a \emph{phase transition} if there exists $\beta$ such that $|\mathcal{G}_{\beta}|>1$. This definition is equivalent to prove that, for the same $\beta$, we have 
$\mu_{\beta,\bar{h}}^+\neq \mu_{\beta,\bar{h}}^-$, where $\mu_{\beta,\bar{h}}^{\pm}=\lim_{n\to +\infty} \mu_{n,\beta,\bar{h}}^{\pm}$ (One can find the proof of this result in \cite{Ge}). We say that the model \emph{has uniqueness at inverse temperature $\beta$} if $|\mathcal{G}_{\beta}|=1$.

We say that a Gibbs measure $\mu$ is \emph{extremal} if $\mu$ cannot be decomposed as a convex combination of other Gibbs measures. We know that $\mu^+_{\beta,\bar{h}}$ and $\mu^{-}_{\beta,\bar{h}}$ are extremal Gibbs measures (see \cite{Ge}).

\subsection{Splitting Gibbs Measures}

Let $h_n \in \mathbb{R}$ for $n\ge 1$ and $b_x\in \mathbb{R}$ for each $x\in V$.
The  $b_x$ will be the \emph{boundary fields}, and often we will write $b_n$, in the situation where the boundary fields only depend on the generation $n$. We define the probability measure $\mu_n$ on the volume $V_n$ as
\begin{equation}\label{gibbs_free}
\mu_{n}(\sigma)=\frac{1}{Z_{n}}\exp\left\{ -\beta H^{0}_{n}(\sigma)+ \sum_{k=1}^{n-1}\sum_{x\in W_{k}}h_k\sigma_x +\sum_{x\in W_n}b_x\sigma_x \right\},
\end{equation}
where $\beta$ is the inverse temperature, and $Z_{n}$ is the partition function given by
\begin{equation}
Z_{n}=\sum_{\sigma_{V_n}\in \Omega_{V_n}}\exp\left\{ -\beta H^{0}_{n}(\sigma) +\sum_{k=1}^{n-1}\sum_{x\in W_{k}}h_k\sigma_x +\sum_{x\in W_n}b_x\sigma_x. \right\}.
\end{equation}

We say that the probability measures $\mu_{n}$ are \emph{compatible} if for all $n\ge 1$ and $\sigma\in \Omega_{V_{n-1}}$:
\begin{equation}\label{compatibility}
\sum_{\omega\in \Omega_{W_n}}\mu_n(\sigma \lor \omega) = \mu_{n-1}(\sigma).
\end{equation}
Here $\sigma \lor \omega$ is the concatenation of the configurations. By Kolmogorov's Theorem, there exists a unique measure $\mu$ on $\Omega$ such that, for all $n$ and $\sigma_{V_n}\in \Omega_{V_n}$,
\begin{equation}
\mu(\sigma|_{V_n}=\sigma_{V_n})=\mu_{n}(\sigma_{V_n}).
\end{equation}
This measure is called a \emph{splitting Gibbs measure} (For Kolmogorov's Theorem, see \cite{Shiryaev}).

By \cite{Ge}, we know that the Gibbs measures with boundary condition plus and minus are extremal, and all extremal Gibbs measures are splitting Gibbs measures. The free-boundary Gibbs measure $\mu^{\sharp}$ at sufficiently low temperatures is not extremal, but still a splitting Gibbs measure. Taking an arbitrary  convex mixture of those three states in general will give a non-splitting measure, however. \footnote{There is an active interest under which conditions the free-boundary measure, and its analogues in external fields, for Potts models etc, are extremal. This is also known as the ``reconstruction problem''. See e.g. \cite{Moss,EKPS,FK}. However, for the present work this question plays no role.}

\section{Compatibility}

From now on, we will assume that $b_x$ for all $x\in V$ is such that $b_{x}=b(\lVert x\rVert)$, i.e., $b_x$ depends only on the distance of $x$ from the root. Writing $b_n=b_x$ when $\lVert x \rVert = n$, we will usually also assume that $b_n$ is decreasing. 
The following result is adapted from Rozikov (\cite{Rozi1} Theorem 2.1), who treated general inhomogeneous fields.

\begin{theorem*}\label{thm_compatibility}
The probability measures $(\mu_n)_{n\ge 1}$ are compatible if and only if for every $n\ge 2$ the following equation holds,
\begin{equation}\label{recurrence}
b_{n-1}=h_{n-1}+d F(b_n,\theta),
\end{equation}  
where $\theta=\tanh(\beta J)$ and $F(x,\theta)=\arctanh(\theta \tanh x)$.
\end{theorem*}

\begin{proof}
Suppose that (\ref{compatibility}) holds. Substituting in the probability measure (\ref{gibbs_free}), we have
$$
\begin{aligned}
&\frac{Z_{n-1}}{Z_n}\sum_{\omega\in \Omega_{W_n}}\exp\left\{ -\beta H^0_{n-1}(\sigma) + \beta J \sum_{x\in W_{n-1}}\sum_{y\in S(x)}\sigma_x\omega_y + \sum_{k=1}^{n-1}\sum_{x\in W_{k}}h_k\sigma_x+b_n\sum_{x\in W_n}\omega_x \right\}\\
&=\exp\left\{ -\beta H^0_{n-1}(\sigma)+ \sum_{k=1}^{n-2}\sum_{x\in W_{k}}h_k\sigma_x +b_{n-1}\sum_{x\in W_{n-1}}\sigma_x \right\}.
\end{aligned}
$$ 
This implies the following equation,
$$
\frac{Z_{n-1}}{Z_n}\sum_{\omega \in \Omega_{W_n}}\exp\left\{\sum_{x\in W_{n-1}}\sum_{y\in S(x)}(\beta J \sigma_x\omega_y+b_n \omega_y) \right\}
=\exp\left\{(b_{n-1}-h_{n-1})\sum_{x\in W_{n-1}}\sigma_x \right\}.
$$
Thus,
$$
\frac{Z_{n-1}}{Z_n}\sum_{\omega \in \Omega_{W_n}}\prod_{x\in W_{n-1}}\prod_{y\in S(x)}\exp\{\beta J \sigma_x\omega_y + b_n\omega_y\}
=\prod_{x\in W_{n-1}}\exp\{(b_{n-1}-h_{n-1})\sigma_x\},
$$
and therefore,
$$
\frac{Z_{n-1}}{Z_n}\prod_{x\in W_{n-1}}\left(\sum_{u\in \{-1,1\}}\exp\{\beta J \sigma_x u + b_n u\}\right)^{d}
=\prod_{x\in W_{n-1}}\exp\{(b_{n-1}-h_{n-1})\sigma_x\}.
$$

Substituting in the last equality  $\sigma_x=1$ and $\sigma_x=-1$ for all $x\in V$, and dividing the first expression by the second one, we obtain
$$
\left( \frac{\sum_{u\in \{-1,1\}}\exp\{\beta J u+b_n u\}}{\sum_{u\in \{-1,1\}}\exp\{-\beta J u+b_n u\}} \right)^{d}=\exp\{2 (b_{n-1}-h_{n-1})\}.
$$
Taking logarithms, we get the desired formula. Note that
$$
F(x,\theta)=\arctanh(\theta \tanh x)=\frac{1}{2}\log\frac{(1+\theta)e^{2x}+(1-\theta)}{(1-\theta)e^{2x}+(1+\theta)}.
$$
For the converse, note that
$$
\begin{aligned}
&\sum_{\omega\in \Omega_{W_n}}\mu_n(\sigma \lor \omega)\\
&=\frac{1}{Z_{n}}\exp\left\{ -\beta H^0_{n-1}(\sigma)+\sum_{k=1}^{n-1}\sum_{x\in W_k}h_k\sigma_x  \right\}\prod_{x\in W_{n-1}}\left( \sum_{u\in \{-1,1\}}\exp\left\{  \beta J \sigma_x u + b_n u  \right\} \right)^{d}.
\end{aligned}
$$
For any $t\in \{-1,1\}$, we have the following identity,
$$
\left( \sum_{u\in \{-1,1\}}\exp\left\{ \beta J t u + b_n u  \right\} \right)^{d}=a_n\exp\{t (b_{n-1}-h_{n-1})\},
$$
for some $a_n>0$. Consider the function $A_n=a_n^{s}$ where $s=|W_{n-1}|$. We have
$$
\begin{aligned}
&\sum_{\omega \in \Omega_{W_n}}\mu_n(\sigma \lor \omega)\\
&=\frac{A_n}{Z_{n}}\exp\left\{ -\beta H^0_{n-1}(\sigma)+\sum_{k=1}^{n-1}\sum_{x\in W_k}h_k\sigma_x \right\} \prod_{x\in W_{n-1}}\exp\{\sigma_x (b_{n-1}-h_{n-1})\}
\end{aligned}
$$
Using the fact that $\sum_{\sigma \in \Omega_{V_{n-1}}}\sum_{\omega\in \Omega_{W_n}}\mu_n(\sigma \lor \omega)=1$, we have $A_n=Z_n/Z_{n-1}$.
\end{proof}

By Theorem \ref{thm_compatibility}, there is a bijection between sets $\mathbf{b}=\{b_n,\ n\ge 1\}$ satisfying  equation (\ref{recurrence}) and  splitting Gibbs measures $\mu$. Thus, in particular, the extremal Gibbs measures $\mu^{\pm}_{\beta,\bar{h}}$ are associated to the boundary fields $\mathbf{\tilde{b}^{\pm}}=\{\tilde{b}^{\pm}_n,\ n\ge 1\}$.

The homogeneous splitting Gibbs measures for the Ising model in homogeneous fields, i.e., $h_n=h$ for all $n\ge 1$, are very well known, see e.g. \cite{Ge} and \cite{Rozi1}. The translation-invariant solutions $(b_n)_{n\ge 1}$ to the recurrence equation (\ref{recurrence}) when $h_n=h$ for all $n\ge 1$ are constant, i.e., $b_n=b^*$ for all $n\ge 1$. Thus, we have the equation 
\begin{equation}\label{fixed_point}
b^*=h+d F(b^*,\theta):=\psi(b^*),
\end{equation}
where $\psi(x)=h+dF(x,\theta)$.

By \cite{Ge}, \cite{Rozi1} we know that there exists $\beta_c(d)>0$ and $h_c(\beta,d)>0$ such that:
\begin{enumerate}
\item[(1)] If $\beta\le \beta_c(d)$ or $|h|> h_c(\beta,d)$, the function $\psi$ has exactly one fixed point. We define the solution as the sequence $\mathbf{b^{\sharp}}=\{b^{\sharp}_n\}_{n\ge 1}$ such that $b^{\sharp}_n=b^{\sharp}$ is constant.
\item[(2)] If $\beta>\beta_c(d)$ and $|h|< h_c(\beta,d)$, the function $\psi$ has exactly three fixed points. The solutions are the sequences $\mathbf{b^{\sharp}}$, $\mathbf{b^+}=\{b^+_n\}_{n\ge 1}$ and $\mathbf{b^-}=\{b^-_n\}_{n\ge 1}$ in which are constant $b^+_n=b^+$ and $b^-_n=b^-$ for all $n\ge 1$.
\item[(3)] If $\beta>\beta_c(d)$ and $|h|= h_c(\beta,d)$, the function $\psi$ has exactly two fixed points. For $h=h_c$ the sequences $\mathbf{b^{\sharp}}$ and $\mathbf{b^-}$ coincide, and for $h=-h_c$ the sequences $\mathbf{b^{\sharp}}$ and $\mathbf{b^+}$ coincide.
\end{enumerate}
The sequences $\mathbf{b^+}$ and $\mathbf{b^-}$ are extremal in the sense that, if $\mathbf{b}=(b_n)_{n\ge 1}$ is a solution to (\ref{fixed_point}), then $b^- \le b_n \le b^+$ for all $n\ge 1$.

Note that $b^+$ is a saddle node of $\psi$ for $h=-h_c$, i.e., $\psi'(b^+)=1$; this means that $\mathbf{b^+}$ is (marginally) stable in a minus field. It attracts higher values, but repels lower ones. The ``stable'' $\mathbf{b^-}$ is an attractor of all initial $b$ below $b^+$.
In fact it attracts exponentially fast, due to the map $\psi(b)-b$ near $b^-$ being contracting.

Also, we note that similarly, in a positive critical field, $b^-$ is the saddle node of $\psi$ for $h=h_c$, and thus $\mathbf{b^-}$ is marginally stable and $\mathbf{b^+}$ is the stable attractor in a plus field.

Lets $\psi(\infty):=\lim_{x\to \infty}\psi(x)<\infty$. It is easy to see that the iteration $\psi^n(\infty)$ converges to $b^+$ as $n$ goes to infinity (see \cite{Ge, Rozi1}). Moreover, $\psi^n(\infty)\ge b^+$ for all $n\ge 1$.

\begin{tikzpicture}[xscale=1.2,yscale=1.7, decoration={
    markings,
    mark=at position 0.5 with {\arrow{>}}}]
\draw [help lines, <->]  (-5,0) -- (5,0);
\draw [help lines, <->] (0,-2.5) -- (0,2.5);
\draw [domain=-5:5, semithick, smooth] plot (\x, {-0.3+3*(1/2)*ln((1+1/2)/(1-1/2))*tanh(\x))});
\draw [domain=-2.5:2] plot (\x, {\x});
\node at (2,2.2) {$y=x$};
\node at (5,1.2) {$\psi=-h_c+dF(\cdot,\theta)$};
\draw [dashed] (0.73,0) -- (0.73,0.73);
\node [below] at (0.73,0) {$b^+$};
\draw [dashed] (0,0.7) -- (0.7,0.7);
\node [left] at (0,0.7) {$b^+$};
\draw [dashed] (-1.79,0) -- (-1.79,-1.8);
\node [above] at (-1.79,0) {$b^-$};
\draw [dashed] (0,-1.79) -- (-1.8,-1.79);
\node [right] at (0,-1.79) {$b^-$};
\draw [postaction={decorate}, >=stealth] (5,1.8) -- (1.8,1.8);
\draw [postaction={decorate}, >=stealth] (1.8,1.8) -- (1.8,1.3);
\draw [dashed] (1.8,1.8) -- (0,1.8);
\node [left] at (0,1.8) {$\psi(\infty)$};
\draw [postaction={decorate}, >=stealth] (1.8,1.3) -- (1.3,1.3);
\draw [postaction={decorate}, >=stealth] (1.3,1.3) -- (1.3,1.13);
\draw [dashed] (1.3,1.3) -- (0,1.3);
\node [left] at (0,1.3) {$\psi^2(\infty)$};
\draw [->, >=stealth] (1.3,1.13) -- (1.13,1.13);
\draw [->, >=stealth] (1.13,1.13) -- (1.13,1.05);
\node [left] at (-0.2,1) {$\vdots$};
\node[align=center, below] at (0,-3)%
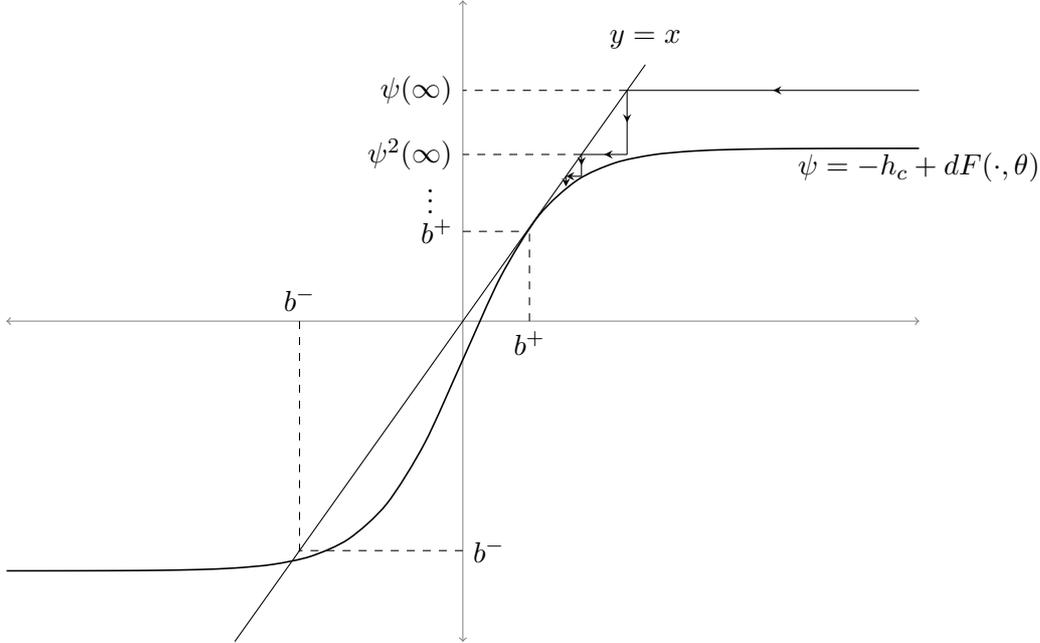
{Figure 1: The graph of $\psi$ for $h=-h_c$ and the fixed points $b^+$ and $b^-$, \\
and the sequence $\psi^n(\infty)$ converging to $b^+$.};
 \end{tikzpicture}

\section{Results and proofs}

Our analysis is based on the behavior of the sum
$
\sum_{j=1}^n\left( \sum_{i=j}^n \epsilon_i \right)^2
$
for the perturbation of the field $(\epsilon_k)_{k\ge 1}$. Firstly, the following inequality in the next proposition is inspired by the rearrangement inequality and somehow helps to see the behavior of the sum.

\begin{proposition}
For any positive decreasing sequence $(\epsilon_n)_{n\ge 1}$,
$$
\sum_{i=1}^n (i\epsilon_i)^2 \le \sum_{j=1}^n\left( \sum_{i=j}^n \epsilon_i \right)^2 \le \sum_{i=1}^n ((n-i+1)\epsilon_i)^2.
$$
\end{proposition}

\begin{proof}
For the upper bound, using that $\epsilon_i\le \epsilon_j$ for any $i\ge j$, we have
\begin{equation}
\sum_{j=1}^n\left( \sum_{i=j}^n \epsilon_i \right)^2 \le \sum_{j=1}^n\left( \sum_{i=j}^n \epsilon_j \right)^2=\sum_{i=1}^n ((n-i+1)\epsilon_i)^2.
\end{equation}
For the lower bound, we use the following expression,
\begin{equation}
\sum_{j=1}^n\left( \sum_{i=j}^n \epsilon_i \right)^2 = \sum_{i=1}^n i\epsilon^2_i +2\sum_{i=2}^n \epsilon_i \left( \sum_{j=1}^{i-1}j\epsilon_j \right).
\end{equation}
Thus,
\begin{equation}
\sum_{j=1}^n\left( \sum_{i=j}^n \epsilon_i \right)^2
\ge \sum_{i=1}^n i\epsilon^2_i +2\sum_{i=2}^n \epsilon_i^2\left( \sum_{j=1}^{i-1}j \right)=\sum_{i=1}^n (i \epsilon_i)^2,
\end{equation}
as we wanted.
\end{proof}

The results will be based on estimate of the influence from infinity on boundary fields near the origin. If the influence due to the inhomogeneous fields is small enough, the non-uniqueness of Gibbs measures will not change, if the influence of the inhomogeneous terms gets too big, then we will have a unique Gibbs measure.

The external fields that we are working with are of the form $h_n=-h_c-\epsilon_n$. Consider the function $\tilde{\psi}_n(x):=h_{n}+dF(x,\theta)=\psi(x)-\epsilon_n$. For each $n\ge 1$, define $\tilde{\psi}_{n,n}(x)=\tilde{\psi}_n(x)$ and $\tilde{\psi}_{k,n}(x)=\tilde{\psi}_{k}(\tilde{\psi}_{k+1,n}(x))$ for $k< n$.
Note that, for any $k\ge 1$,
\begin{equation}\label{limit_plus}
\tilde{b}^+_k = \lim_{n\to \infty}\tilde{\psi}_{k,n}(\tilde{b}^+_{n+1})\ge \lim_{n\to \infty}\tilde{\psi}_{k,n}(b^+),
\end{equation}
since the sequence $c_k=\lim_{n\to \infty}\tilde{\psi}_{k,n}(b^+)$ is well defined, and satisfies the compatibility recurrence (\ref{recurrence}). Thus, by the extremality of $\tilde{\mathbf{b}}^+$, the inequality (\ref{limit_plus}) holds. By the same argument, for any $k\ge 1$,
\begin{equation}
\tilde{b}^-_k = \lim_{n\to \infty}\tilde{\psi}_{k,n}(\tilde{b}^-_{n+1})\le  \lim_{n\to \infty}\tilde{\psi}_{k,n}(b^-).
\end{equation}

\begin{theorem}\label{phase_transition}
Consider the ferromagnetic Ising model on a Cayley tree $\Gamma^d$ with external fields $(-h_c-\epsilon_n)_{n\ge 1}$. Suppose that the sequence of positive $(\epsilon_n)_{n\ge 1}$ decreases to zero and satisfies the following condition,
\begin{equation}\label{weak}
\lim_{n\to \infty}\sum_{j=1}^n\left( \sum_{i=j}^n \epsilon_i \right)^2<\infty.
\end{equation}
Then the perturbed model undergoes a phase transition.
\end{theorem}

\begin{proof}
Suppose that $\theta>1/d$, which implies that we have phase transition and let $b^-<b^+$ be the solutions of the equation (\ref{fixed_point}). Remind that $b^+$ is a saddle node, i.e., $\psi'(b^+)=1$.

By Taylor expansion on $\psi$, we have
\begin{equation}
\psi(b^{+}-\epsilon_n)=\psi(b^+)-\psi'(b^+)\epsilon_n+\frac{1}{2}\psi''(b^+)\epsilon_n^2+O(\epsilon_n)^3.
\end{equation}
Using the fact that $\psi$ is a concave function around $b^+$, we get
$$
\psi(b^{+}-\epsilon_n)=b^{+} -\epsilon_n -\frac{1}{2}|\psi''(b^+)|\epsilon_n^2+O(\epsilon_n)^3.
$$
Now, if we apply $\psi$ on  $\psi(b^{+}-\epsilon_n)-\epsilon_{n-1}$, we have
$$
\begin{aligned}
&\psi(\psi(b^{+}-\epsilon_n)-\epsilon_{n-1})=\psi(b^{+} - \epsilon_n-\epsilon_{n-1} -\frac{1}{2}|\psi''(b^+)|\epsilon_n^2+O(\epsilon_n)^3)\\
&=b^{+}-(\epsilon_n+\epsilon_{n-1}) -\frac{1}{2}|\psi''(b^+)|\epsilon_n^2-\frac{1}{2}|\psi''(b^+)|(\epsilon_n+\epsilon_{n-1} +\frac{1}{2}|\psi''(b^+)|\epsilon_n^2)^2+
O(\epsilon_{n-1})^3\\
&=b^{+}-(\epsilon_n+\epsilon_{n-1}) -\frac{1}{2}|\psi''(b^+)|\left(\epsilon_n^2+(\epsilon_n+\epsilon_{n-1})^2\right)+O(\epsilon_{n-1})^3.\\
\end{aligned}
$$
Thus, by induction, we obtain our main formula: 
\begin{equation}\label{psi_iteracted}
\tilde{\psi}_{k,n}(b^{+})=b^{+}- \sum_{i=k}^n \epsilon_i-\frac{1}{2}|\psi''(b^+)|\sum_{i=k+1}^n\left(\sum_{j=i}^n\epsilon_j \right)^2+O(\epsilon_{k+1})^3.
\end{equation}
Since the sequence $(\epsilon_n)_{n\ge 1}$ satisfies (\ref{weak}), for any $\varepsilon>0$, there exists $k_1\ge 1$ such that, for all $k\ge k_1$,
\begin{equation}
\sum_{j=k+1}^n\left( \sum_{i=j}^n \epsilon_i \right)^2<\frac{\varepsilon}{|\psi''(b^+)|},
\end{equation}
for every $n\ge k$. Moreover, $(\epsilon_n)_{n\ge 1}$ is summable. Thus, there exists $k_2\ge 1$ such that, for all $k\ge k_2$,
\begin{equation}\label{summable}
\sum_{i=k}^n \epsilon_i < \frac{\varepsilon}{2},
\end{equation}
for every $n\ge k$. Thus, for all $k\ge \max\{k_1,k_2\}$,
\begin{equation}\label{small}
\tilde{\psi}_{k,n}(b^{+})>b^{+}- \varepsilon +O(\epsilon_{k+1})^3.
\end{equation}
for every $n\ge k$. Let's take $\varepsilon>0$ sufficiently small such that $b^+ - \varepsilon >b^-$.
From (\ref{psi_iteracted}), there exists $k_0\ge \max\{k_1,k_2\}$ and $\delta=\delta(k)>0$  such that, for all $k\ge k_0$,
\begin{equation}
\tilde{\psi}_{k,n}(b^+) > b^{+}-\delta,
\end{equation}
for every $n> k$. Thus, by (\ref{limit_plus}), $\tilde{b}^+_{k}> b^+ - \delta>b^-$ for all $k\ge k_0$. Note also that $b^-\ge \tilde{b}^-_{k}$, as $b^-$ is even stable for perturbation under homogeneous fields, so much the  more it is for $\mathbf{\tilde{b}^-}$. Therefore $\tilde{b}_{k}^+>\tilde{b}_{k}^-$ for all $k\ge k_0$. Since $\tilde{b}^+$ and $\tilde{b}^-$ are associated to extremal Gibbs measures, namely the Gibbs measures with plus and minus-boundary condition, these measures are distinct.
\end{proof}

It is easy to see that the above result also works when we will consider the stability of the minus state in a plus field, under addition of a positive spatially dependent perturbation $(\epsilon_n)_{n\ge 1}$. In that case, we consider the ferromagnetic Ising model on a Cayley tree $\Gamma^d$ with external fields $(h_c+\epsilon_n)_{n\ge 1}$.

\smallskip

Our second result says that if the inhomogeneous field is negative and such that it does {\em not} satisfy the above condition, it is strong enough to remove the phase transition, and indeed there will be one single Gibbs measure.

\begin{theorem}\label{uniqueness}
Consider the ferromagnetic Ising model on a Cayley tree $\Gamma^d$ with external fields $(-h_c-\epsilon_n)_{n\ge 1}$. Suppose that the sequence of positive $(\epsilon_n)_{n\ge 1}$ decreases to zero and satisfies the following condition,
\begin{equation}\label{strong}
\lim_{n\to \infty}\sum_{j=1}^n\left( \sum_{i=j}^n \epsilon_i \right)^2=\infty.
\end{equation}
Then the perturbed model has uniqueness for any temperature.
\end{theorem}

\begin{proof}
For $1\le k<n<N$, let us consider the auxiliary boundary fields $(b^{+,k,n,N}_m)_{1\le m\le N}$ in which $b^{+,k,n,N}_{N}=+\infty$, and $b^{+,k,n,N}_{m-1}=\psi(b^{+,k,n,N}_{m})$ for $n<m\le N$, and $b^{+,k,n,N}_{m-1}=\tilde{\psi}_{m-1}(b^{+,k,n,N}_{m})$ for $k<m\le n$, and $b^{+,k,n,N}_{m}$ satisfies the compatibility equation (\ref{compatibility}) for $m\le k$. This sequence means that we are taking plus boundary condition at distance $N$ from the origin, between $n$ and $N$ the sequence is in the homogeneous case, and between $k$ and $n$ the sequence is in the inhomogeneous case. Note that this provides us with an upper bound for $\mathbf{\tilde{b}^+}$.

\begin{center}
\begin{tikzpicture}
\node at (0,0) {$\cdot$};
\node [left] at (0,0) {$O$};
\node [above] at (1.2,0) {$k$};
\node [above] at (2.7,0) {$n$};
\node [above] at (3.7,0) {$N$};
\node [below] at (3.8,0) {\tiny{$+\infty$}};
\node [below] at (3,0) {\tiny{hom.}};
\node [below] at (1.7,0) {\tiny{non-hom.}};
\node [below] at (0.4,0) {\tiny{compat.}};
\node [below] at (0,-1.3) {$-h_c-\epsilon_i$};
\node [below] at (0,-2.7) {$-h_c$};
\tkzDefPoint(0,0){O}
\tkzDefPoint(0,1){A1}
\tkzDefPoint(0,2.5){A2}
\tkzDefPoint(0,3.5){A3}
\tkzDrawCircle(O,A1)
\tkzDrawCircle(O,A2)
\tkzDrawCircle(O,A3)
\draw [help lines, <->, black]  (0.05,0) -- (0.95,0);
\draw [help lines, <->, black]  (1.05,0) -- (2.45,0);
\draw [help lines, <->, black]  (2.55,0) -- (3.45,0);
\node[align=center, below] at (0,-4)%
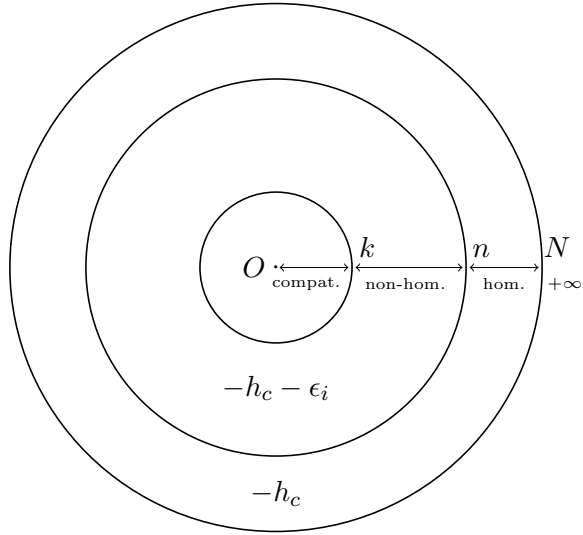
{Figure 2: The Cayley tree with root $O$, and the auxiliary boundary fields.\\ 
The circle means the depth of the tree.};
\end{tikzpicture}
\end{center}

In order for all Taylor expansions in (\ref{psi_iteracted}) to hold, for a fixed $\varepsilon>0$, there exist $N\ge 1$ sufficiently large and $n\le N$ such that $b^+ < b^{+,k,n,N}_n < b^+ + \varepsilon$, since $\psi^p(\infty)$ converges to $b^+$ as $p\to \infty$. Note that, for a larger $N$, we can increase $n$ as well. Thus, taking the limit in $N$, we have $\lim_{N\to \infty}b^{+,k,n,N}_n=b^+$. So, we can consider the sequence $(b^{+,k,n}_m)_{m\ge 1}$ defined on the whole Cayley tree by $b^{+,k,n}_m=b^+$ for $m\ge n$, and $b^{+,k,n}_{m-1}=\tilde{\psi}_{m-1}(b^{+,k,n}_{m})$ for $1<m\le n$. For a fixed $0<\delta<b^+$, there exists $k\le n$ such that
\begin{equation}\label{divergence}
\sum_{j=k+1}^n\left( \sum_{i=j}^n \epsilon_i \right)^2 > \frac{2\delta}{|\psi''(b^+)|}.
\end{equation}
By (\ref{psi_iteracted}),
\begin{equation}
\tilde{\psi}_{k,n}(b^+)=b^+ - \sum_{i=k}^n \epsilon_i - \frac{1}{2}|\psi''(b^+)|\sum_{j=k+1}^n\left( \sum_{i=j}^n \epsilon_i \right)^2 +O(\epsilon_{k+1})^3,
\end{equation}
we thus have $\tilde{\psi}_{k,n}(b^+) < b^+ -\delta +O(\epsilon_{k+1})^3 +O(\varepsilon)$. Let us choose $k$ very large (and so $n$) such that $\epsilon_k>0$ is small enough, in which case there exists $0<\delta'<b^+$ satisfying
\begin{equation}\label{b}
b^{+,k,n}_k< b^+ -\delta'.
\end{equation}
Note that $\delta'$ does not depend on $k$ in the sense that $\delta'$ does not change when $k$ increases, once $n$ is large enough, satisfying (\ref{divergence}). The bound (\ref{b}) means that the perturbation of the external fields $(\epsilon_k)_{k\ge 1}$ is strong enough in the sense of condition (\ref{strong}) for the boundary field $b^{+,k,n}_m$ be strictly below $b^+$ for every $m\le k$. Moreover, we take $k$ sufficiently large so that there exists $\delta''<1$ such that $0<\psi'(b^{+,k,n}_m)<1-\delta''$ for $m\le k$. Thus, for every $k_0\ge 1$, there exist $k>k_0$ and $n_k>k$ such that (\ref{b}) holds. Note that $\tilde{b}^+_{k_0}=\lim_{k\to \infty}b^{+,k,n_k}_{k_0}$.

Now, define the boundary fields $(b^{-,n,N}_m)_{1\le m\le N}$ in which $b^{-,n,N}_{N}=-\infty$, $b^{-,n,N}_{m-1}=\psi_{m-1}(b^{-,n,N}_{m})$ for $n< m\le N$ and  $b^{-,n,N}_{m-1}=\tilde{\psi}_{m-1}(b^{-,n,N}_{m})$ for $1\le m\le  n$. For a fixed $n<N$, taking $N$ to infinity, we have $\lim_{N\to \infty}b^{-,N}_n =b^+$. Thus, we can consider the sequence $(b^{-,n}_m)_{m\ge 1}$ defined in the whole Cayley tree so that $b^{-,n}_m=b^-$ for all $m\ge n$, and $b^{-,n}_{m-1}=\tilde{\psi}_{m-1}(b^{-,n}_{m})$ for $1\le m\le  n$.
Note that $\tilde{b}^-_{k_0}=\lim_{n\to \infty}b^{-,n}_{k_0}=\lim_{n\to \infty}\tilde{\psi}_{k_0,n-1}(b^-)$ for all $k_0\ge 1$. The Taylor expansion below shows that these boundary fields decay exponentially by contraction, since $0<\psi'(b^-)<1$. In fact,
\begin{equation}\label{psi_iteracted_minus2}
\tilde{\psi}_{k,n-1}(b^{-})=b^{-}- \sum_{i=k}^{n-1} \psi'(b^-)^{i-k}\epsilon_i +O(\epsilon_{k+1})^2.
\end{equation}

Since the map $\psi$, and similarly $\tilde{\psi}_m$, act as contractions on the interval $(- \infty, b^{+}-\delta')$, with a uniform contraction bound $1- \delta''$, we have that in the limit $k$ to infinity the difference in influence on the boundary field at sites $k_0$ between negative boundary fields and positive boundary fields less than $b^{+}-\delta'$ at sites at distance $k_0$ disappears. In fact, we know that the sequence $b^{-,n}_{k}$ converges to $b^-$ as $k$ is going to infinity (and so $n$ goes to infinity), since $b^{-,n_k}_k \le b^{+,k,n_k}_{k}\le b^+-\delta'$, we have that $| b^{+,k,n_k}_{k} - b^{-,n_k}_k |<C$ for some $C>0$. Note that $C$ does not depend on $k$. Thus, by the Mean Value Theorem,
\begin{equation}
\begin{aligned}
\tilde{b}^+_{k_0}-\tilde{b}^-_{k_0}&= \lim_{k\to \infty}b^{+,k,n_k}_{k_0} - \lim_{n\to \infty}b^{-,n}_{k_0} \\
&= \lim_{k\to \infty} \tilde{\psi}_{k_0,k-1}(b^{+,k,n_k}_{k}) - \lim_{k\to \infty}\tilde{\psi}_{k_0,k-1}(b^{-,n_k}_k) \\
&=\lim_{k\to \infty} \left( \tilde{\psi}_{k_0,k-1}(b^{+,k,n_k}_{k}) - \tilde{\psi}_{k_0,k-1}(b^{-,n_k}_k) \right) \\
&\le \lim_{k\to \infty} \left( \sup_{c\in [b^{-,n_k}_k,b^{+,k,n_k}_{k}]}\tilde{\psi}'_{k_0,k-1}(c) | b^{+,k,n_k}_{k} - b^{-,n_k}_k | \right) \\
&\le C\lim_{k\to \infty}(1-\delta'')^{k-k_0}\\
&=0.
\end{aligned}
\end{equation}
Thus, the sequences $\mathbf{\tilde{b}^+}$ and $\mathbf{\tilde{b}^-}$ are equal. Therefore the extremal Gibbs measures associated to these sequences, $\mu^+_{\beta,\bar{h}}$ and $\mu^-_{\beta,\bar{h}}$ respectively, are equal for any $\beta>0$.
\end{proof}

\textbf{Example. }As in \cite{BCCP, CV}, let us consider the sequence $\epsilon_k=1/k^{\gamma}$, where $\gamma>0$. Note that
\begin{equation}\label{gamma}
\sum_{j=1}^n\left( \sum_{i=j}^n i^{-\gamma} \right)^2 = \sum_{j=1}^n O(j^{2-2\gamma}).
\end{equation}
Thus, the sum (\ref{gamma}) converges when $2-2\gamma<-1$ and diverges when $2-2\gamma\ge-1$, and we conclude that the critical power is $\gamma_c=3/2$. Note also that the model has uniqueness at the critical power, since the sum diverges.

\section*{Acknowledgments}

A.C.D.van Enter thanks Henk Broer for a helpful conversation. We thank Bruno Kimura and Wioletta Ruszel for discussions and for the support and hospitality in Delft, making this collaboration possible. R. Bissacot thanks Wioletta especially  for organising the invitation, support and all the effort for his first visit in the Netherlands. We thanks NWO-STAR and FAPESP for support.

\end{document}